\documentclass[11pt,letterpaper]{article}

\usepackage[margin=1in]{geometry}
\usepackage{amsmath,amsthm,amsfonts,amssymb,amscd}
\usepackage{lastpage}
\usepackage{enumerate}
\usepackage{fancyhdr}
\usepackage{mathrsfs}
\usepackage{xcolor}
\usepackage{graphicx}
\usepackage{listings}
\usepackage[colorlinks=true, urlcolor=blue, linkcolor=blue, citecolor=magenta]{hyperref}
\usepackage{cite}
\usepackage{multirow}
\usepackage{comment}
\usepackage{indentfirst}
\usepackage{todonotes}
\usepackage{cleveref}

\hypersetup{%
  colorlinks=true,
  linkcolor=red,
  linkbordercolor={0 0 1}
}
 \newcommand{\vnote}[1]{{\color{blue} [{\small {#1}} --Venkat]}}

\lstdefinestyle{Matlab}{
    language        = matlab,
    frame           = lines, 
    basicstyle      = \footnotesize,
    keywordstyle    = \color{blue},
    stringstyle     = \color{green},
    commentstyle    = \color{red}\ttfamily
}
\newtheorem{theorem}{Theorem}
\newtheorem{definition}[theorem]{Definition}

\newtheorem{question}{Open Question}

\newtheorem{lemma}[theorem]{Lemma}
\newtheorem{hypothesis}[theorem]{Hypothesis}
\newtheorem{corollary}[theorem]{Corollary}

\newtheorem{example}[theorem]{Example}

\newtheorem{conjecture}[theorem]{Conjecture}

\newcommand{\wone}{\mathsf{W[1]}}
\newcommand{\mone}{\mathsf{M[1]}}
\newcommand{\fpt}{\mathsf{FPT}}

\newcommand{\kclique}{\textsc{$k$-Clique}}
\newcommand{\kexact}{\textsc{$k$-ExactCover}}
\newcommand{\ksetcover}{\textsc{$k$-SetCover}}

\newcommand{\FormatAuthor}[3]{
\begin{tabular}{c}
#1 \\ {\small\texttt{#2}} \\ {\small #3}
\end{tabular}
}

\title{Baby PIH: Parameterized Inapproximability of Min CSP}

\author{Venkatesan Guruswami\thanks{Simons Institute for the Theory of Computing, and Departments of EECS and Mathematics, UC Berkeley. Research supported in part by NSF grants CCF-2228287 and CCF-2211972 and a Simons Investigator award. Email: {\tt venkatg@berkeley.edu}} \and
Xuandi Ren\thanks{Department of EECS, UC Berkeley. Supported in part by NSF CCF-2228287. Email: {\tt xuandi\_ren@berkeley.edu}}
\and 
Sai Sandeep\thanks{Work done while at UC Berkeley, supported in part by NSF CCF-2228287. Email: {\tt saisandeep@berkeley.edu}}
}

\date{}

\begin{document}

\maketitle
\thispagestyle{empty}

\begin{abstract}
The Parameterized Inapproximability Hypothesis (PIH) is the analog of the PCP theorem in the world of parameterized complexity. It asserts that no  FPT algorithm can distinguish a satisfiable 2CSP instance from one which is only $(1-\varepsilon)$-satisfiable (where the parameter is the number of variables) for some constant $0<\varepsilon<1$. 

\smallskip
We consider a minimization version of CSPs (Min-CSP), where one may assign $r$ values to each variable, and the goal is to ensure that every constraint is satisfied by some choice among the $r \times r$ pairs of values assigned to its variables (call such a CSP instance $r$-list-satisfiable). We prove the following strong parameterized inapproximability for Min CSP: For every $r \ge 1$, it is $\wone$-hard to tell if a 2CSP instance is satisfiable or is not even $r$-list-satisfiable. We refer to this statement as ``Baby PIH", following the recently proved Baby PCP Theorem (Barto and Kozik, 2021). Our proof adapts the combinatorial arguments underlying the Baby PCP theorem, overcoming some basic obstacles that arise in the parameterized setting. Furthermore, 
our reduction runs in time polynomially bounded in both the number of variables and the alphabet size, and thus implies the Baby PCP theorem as well.
\end{abstract}

\section{Introduction}

\noindent Approximation algorithms and fixed parameter-tractabililty (FPT) are two ways to cope with NP-hard problems. Recently, there have been many works that unite the two by obtaining approximation algorithms for NP-Hard problems that run in FPT time. Examples include {\sc Vertex Coloring}\cite{DHK05,Mar08}, {\sc $k$-Path Deletion}\cite{Lee19}, {\sc Vertex Cycle Packing}\cite{Lok+17}, {\sc Flow Time Scheduling}\cite{Wie18}, {\sc Max $k$-Vertex Cover} in $d$-uniform hypergraphs \cite{SF17,Man19}, {\sc $k$-Means} and {\sc $k$-Median} \cite{LiS16,BPR+17,KMN+04,ANSW20,CGTS02,CGK+19}. 
On the other hand, there are also various developments in FPT hardness of approximation, for example, for {\sc $k$-Biclique} \cite{Lin18}, \kclique{} \cite{CCK+17, Lin21, LRSW22, KK22, LRSW23, CFL+23}, \ksetcover{} \cite{CCK+17,CL19,KLM19,Lin19,KN21,LRSW23a} and so on.
We refer to the survey by Feldmann, Karthik, Lee, and Manurangsi~\cite{FKLM20} for a more comprehensive list of both FPT approximation algorithms and FPT hardness of approximation results.

However, it is worth pointing out that the techniques used in proving FPT hardness of approximation results have been rather problem-specific, and there remain other basic problems for which even constant inapproximability is unknown. This situation is due to the lack of a PCP-like theorem in the FPT world. In classical polynomial time hardness of approximation, the PCP theorem is a fundamental result that leads to a plethora of other results via reductions. The analog of the PCP theorem in the FPT regime was explicitly proposed by Lokshtanov et al~\cite{LRSZ20} and named \emph{Parameterized Inapproximability Hypothesis (PIH)}. The PIH states that there is an absolute constant $\varepsilon > 0$ such that for a 2CSP on $k$ variables with alphabet size $n$, there is no algorithm running in time $f(k)\cdot n^{O(1)}$ that can distinguish a satisfiable instance from one where at most $(1-\varepsilon)$ fraction of constraints can be simultaneously satisfied.

Analogous to the PCP theorem, PIH not only implies many FPT time inapproximability results including \kclique{}, \ksetcover{}, \kexact{}, and {\sc $k$-SetPacking}, but also unifies the inapproximability landscape by serving as a versatile problem-agnostic starting point.  Previously, PIH was known to be implied by Gap-ETH \cite{DM18, CFM17}, a strong assumption with an inherent gap in it. Establishing PIH under a gap-free hypothesis has been a significant open problem in parameterized complexity. Very recently, following the posting of this work, PIH was established under ETH \cite{GLR+23} using the technique of parallelization. The authors first reduce 3SAT to a vectorized problem called Vector-Valued CSP, and then design parallel probabilistically checkable proofs to verify its satisfiability. However, the proof in \cite{GLR+23} fails to get any inapproximability of 2CSP under $\wone \neq \fpt$, since it is not known how to reduce a $\wone$-complete problem, say $\kclique{}$, to their Vector-Valued CSP. In fact, it was pointed out in \cite{GLR+23} that Vector-Valued CSP is likely to be $\mone$-complete and thus not likely to be $\wone$-hard, where $\mone$ is an intermediate complexity class between $\fpt$ and $\wone$ \cite{fellows2003blow, chen2007isomorphism}. Proving PIH under $\wone \neq \fpt$, therefore, still remains an important open problem.

In this work, we prove a list version of the PIH, which we call \emph{Baby PIH}, under the minimal hypothesis $\wone \neq \fpt$. In Baby PIH, we study $r$-list assignments to the underlying 2CSP instance. An $r$-list assignment has a list $L(u)$ of at most $r$ values to each variable $u$, and a constraint between variables $u,v$ is said to be satisfied by the $r$-list assignments if there is at least one pair of values $x \in L(u), y \in L(v)$ such that $(x,y)$ satisfies the constraint between the variables $u,v$. 

\begin{definition}[Baby PIH]
For all constants $r, C \ge 1$ and every computable function $f$, there is no algorithm running in time $f(k)\cdot n^C$ that can distinguish between the following cases, on input a 2CSP instance on $k$ variables with alphabet size $n$. 
    \begin{itemize}
        \item (Completeness) There is an assignment satisfying all the constraints. 
        \item (Soundness) No $r$-list assignment satisfies all the constraints. 
    \end{itemize}
\end{definition}

The Baby PIH can be compared with the ``Baby PCP theorem,'' which was established via a purely combinatorial gap theorem in a remarkable recent work of Barto and Kozik~\cite{BK22} (who also coined the term Baby PCP). 
The Baby PCP theorem asserts the NP-completeness of distinguishing satisfiable instances of a 2CSP from those which lack a satisfying $r$-list assignment, for all constants $r > 1$.   Baby PCP differs from Baby PIH in the sense that it concerns 2CSP on $n$ variables with constant alphabet size.  The concept of Baby PCP is itself not new and was studied in the early PCP days under the guise of a certain minimization version of Label Cover~\cite{ABSS97} (see also \cite{DS04}).\footnote{Label Cover is stronger than Baby PCP as stated above, as it imposes that the 2CSP relations are functions (usually referred to as projection property in PCP parlance). However, Barto and Kozik's version also has this projection property.} 

The term ``Baby'' stems from the fact that after reducing the soundness parameter to an arbitrarily small positive constant via parallel repetition, one can deduce the Baby PCP theorem from the PCP theorem.  This strategy also works in the PIH setting. Namely, we can use the canonical ``clause-variable'' construction to build a 2CSP instance with projection property, then apply parallel repetition \cite{Rao11} to amplify the soundness to below $\frac{1}{r^2}$. Suppose a 2CSP instance is $r$-list satisfiable, then by randomly picking a value from each variable's list, we can conclude there is a single assignment that satisfies at least $\frac{1}{r^2}$ fraction of constraints. Therefore, being unable to approximate 2CSP within a $\frac{1}{r^2}$ factor implies that it is hard to distinguish whether a instance is satisfiable or not $r$-list satisfiable. In other words, PIH implies Baby PIH. Thus establishing the Baby PIH is a necessary step towards proving the PIH itself, and we also believe that it is a valuable intermediate step. 

In this work, we prove Baby PIH under the minimal complexity assumption $\wone\neq\fpt$. 

\begin{theorem}
\label{thm:main-intro}
    Assuming $\wone\neq\fpt$, Baby PIH is true.
\end{theorem}

Our proof of Theorem~\ref{thm:main-intro} is combinatorial and follows the framework of Barto and Kozik's proof of the Baby PCP theorem. Specifically, given a 2CSP instance with variable set $X$ and a list size $r$, we choose large enough integers $a,b$ depending on $r$, and construct a bipartite direct product 2CSP instance, whose variable set is $\binom{X}{a}\cup\binom{X}{b}$. Given that the product instance is $r$-list satisfiable, we can repeatedly choose smaller integers $a' \le a, b' \le b$ and extract assignments for the instance with variable set $\binom{X}{a'}\cup \binom{X}{b'}$. The new assignments still list satisfy the smaller instance, but the size of each list on one side is decreased by 1, which helps us to do induction.

We highlight that although this proof strategy looks simple, there are basic obstacles to employ it in the Baby PIH setting (compared to the Baby PCP setting). In the PCP world, the alphabet size $|\Sigma|$ is at most some constant. Thus, to extract assignments for the smaller instance, it is affordable to pick $a,b$ large enough depending on $|\Sigma|$. The running time of there reduction is therefore $|X|^{O_{r,|\Sigma|}(1)}$. However in the PIH case, $|\Sigma|$ can be as large as the input size, and the running time of the reduction can only be $f(|X|)\cdot |\Sigma|^{O_r(1)}$ for some computable function $f$. To overcome this barrier, we non-trivially use the underlying structure of $r$-list satisfying assignments. We further note that our reduction runs in time $|X|^{O_r(1)}$, which is polynomial also in $|X|$. Therefore, our methods give a \emph{unified proof of both the Baby PCP theorem and Baby PIH.} \footnote{Barto and Kozik \cite{BK22} derive Baby PCP with the stronger projection property. Using our techniques, we can get Baby PCP or Baby PIH with rectangular constraints, which is a slightly weaker property but still enough for many downstream reductions.}

As we mentioned earlier, PIH implies Baby PIH, and thus our result can be viewed as a first step towards proving the former. An intermediate problem between them is the following average version of Baby PIH: let an $r$-average list assignment be a labeling $L(u)$ for each variable $u$ such that the average cardinality of $L(u)$ over all the variables $u$ is at most $r$. 
\begin{conjecture}[Average Baby PIH]
For any constants $r >1,C$ and any computable function $f$, no algorithm can given as input a 2CSP instance on $k$ variables with size $n$, distinguish between the following two cases in $f(k)\cdot n^{C}$ time:
    \begin{itemize}
        \item (Completeness) There is an assignment satisfying all the constraints.
        \item (Soundness) No $r$-average list assignment satisfies all the constraints.
    \end{itemize}
\end{conjecture}

Note that the difference between the Baby and Average Baby versions is to use the $\ell_\infty$ vs. $\ell_1$ norms of the number of values given to the variables.
Once again, Average Baby PIH has a counterpart in the PCP world, namely the minimization version of Label Cover with $\ell_1$ total cost, which fueled early inapproximability results for {\sc SetCover} and basic problems that concern linear codes and lattices~\cite{ABSS97}, as surveyed in \cite{AL96}. 

Average Baby PIH is an intriguing open problem and could help in making further progress towards PIH. Furthermore, the Average Baby PIH is strong enough to imply some non-trivial inapproximability results. 
Notably, with an additional property called rectangular constraints (which we will define later), it implies constant inapproximability of \kexact{}, which we previously only knew under PIH \cite{Man20}. 

Towards a better understanding of Average Baby PIH, we give a counterexample showing that the direct product construction we use to prove Baby PIH as is cannot establish the average version. This suggests in order to get Average Baby PIH or full PIH, one may need other techniques or constructions. As a candidate, we pose a question that whether the $\wone$-hardness of approximating \kclique{} can help us bypass this counterexample. Please refer to \Cref{sec:discuss} for details. 

\medskip \noindent \textbf{Organization.}
In Section \ref{sec:pre}, we introduce some preliminaries, including the problems considered in this paper and related complexity hypotheses. In Section \ref{sec:baby}, we prove Baby PIH via the direct product construction. We then discuss Average Baby PIH in Section \ref{sec:avg}, and conclude with some open problems in \Cref{sec:discuss}.

\section{Preliminaries}
\label{sec:pre}
We first start by formally defining $2$-CSP.
\begin{definition}[2CSP]\label{def:2csp}
An instance of arity-2 constraint satisfaction problem (2CSP) is a tuple $\Pi=(X,\Sigma,\Phi)$, where:
\begin{itemize}
\item $X=\{x_1,\ldots,x_k\}$ is the set of variables;
\item $\Sigma=\{\Sigma_{x_1},\ldots,\Sigma_{x_k}\}$ is the set of their respective domains of values. We use $|\Sigma|$ to denote the maximum size of any domain, and call it the alphabet size of $\Pi$. 
\item $\Phi=\{\phi_1,\ldots,\phi_m\}$ is the set of constraints. Each constraint $\phi_i$ is 
 ordered tuple $\langle w_i,R_i\rangle$, where $w_i=(w_{i,1},w_{i,2}) \in X^2$ is a pair of variables, and $R_i$ is a 2-ary relation on $\Sigma_{w_{i,1}}$ and $\Sigma_{w_{i,2}}$. 
\end{itemize}

An assignment of the 2CSP instance is a function from all variables in $X$ to their respective domains. A constraint $\phi_i$ is said to be satisfied by an assignment $\sigma$ if $(\sigma(w_{i,1}),\sigma(w_{i,2}))\in R_i$. We say an assignment $\sigma$ satisfies the 2CSP instance if all constraints are satisfied by $\sigma$. 
\end{definition}

For each constraint $\phi_i$, we can without loss of generality assume $w_{i,1} \neq w_{i,2}$, since unary constraints can be removed by restricting the domain of that variable.

We define $r$-list satisfiability, which generalizes the above satisfiability. 
\begin{definition}[$r$-List Satisfiability]
    Given a 2CSP instance $\Pi=(X,\Sigma,\Phi)$, a multi-assignment is a function mapping each variable to a subset of its domain. We define the size of a multi-assignment $\sigma$ as $\max_{x \in X} |\sigma(x)|$.
    
    We say a multi-assignment $\sigma$ $r$-list satisfies $\Pi$ if $\sigma$ is of size at most $r$, and for every constraint $\phi_i$, there exists a pair of values $u \in \sigma(w_{i,1})$ and $v \in \sigma(w_{i,2})$, such that $(u,v) \in R_i$.
\end{definition}

Normal satisfiability can be viewed as $1$-list satisfiability. Note that as $r$ increases, it becomes easier to $r$-list satisfy a 2CSP instance. 

The relations in the $2$CSPs that we construct using the direct product (see \Cref{def:prod}) satisfy a useful structural property, namely, rectangular relations. 
\begin{definition}[Rectangular Relation]\label{def:rect}
    A relation $R \subseteq A \times B$ is said to be rectangular if there is a set $C$ and functions $\pi: A \to C$ and $\sigma: B \to C$ such that $(a,b) \in R$ if and only if $\pi(a)=\sigma(b)$. Equivalently, $R$ is rectangular if for all $a,a' \in A$ and $b,b' \in B$ such that $(a,b)\in R, (a,b')\in R$, and $(a',b) \in R$, we have $(a',b')\in R$.
\end{definition}

Rectangular relations can be informally viewed as consistency checks, and they are often satisfied by $2$CSPs in product constructions. Projection relation, a stronger version of rectangular relation, is ubiquitous in PCP-based hardness reductions.

We now formally define the direct product construction that we use in our proof. 
\begin{definition}[Direct Product Construction]\label{def:prod}
    Given a 2CSP instance $\Pi=(X,\Sigma,\Phi)$, its $t$-wise direct product, denoted as $\Pi^{\odot t}$, is the following 2CSP instance $(X',\Sigma',\Phi')$:
    \begin{itemize}
        \item $X'=\binom{X}{t}$, where each variable is a $t$-sized subset of variables in $\Pi$. 
        \item The domain of each variable $S \in X'$ is the set of all partial satisfying assignments for $S$ in $\Pi$, i.e., all function $\sigma$ that maps each $x \in S$ to its domain in $\Pi$, such that all constraints in $\Pi$ induced by $S$ are satisfied.
        \item $\Phi'$ has of a consistency constraint for each pair of distinct variables in $X'$. For $S,T \in X'$, the assignments $\sigma_S,\sigma_T$  satisfy the constraint if and only if they are consistent on the values for variables in $S \cap T$.
    \end{itemize}
\end{definition}

Our results are based on the hypothesis $\wone \neq \fpt$, which is closely related to  \kclique{}, a fundamental problem in parameterized complexity theory. 
\begin{definition}[\kclique{}]
    An instance of (multicolored) \kclique{} problem is an undirected graph $G=(V=V_1\dot\cup\ldots\dot\cup V_k,E)$, where each $V_i$ is an independent set. The goal is to decide whether we can find $v_1 \in V_1,\ldots,v_k \in V_k$ which form a clique. For $r>1$, the $r$-gap version of \kclique{} asks to distinguish between the following two cases:
    \begin{itemize}
        \item (Yes) There exists $v_1 \in V_1,\ldots,v_k \in V_k$ which form a clique.
        \item (No) The maximum clique in $G$ has size at most $k/r$.
    \end{itemize}
\end{definition}

The $\wone\neq\fpt$ hypothesis states that for any computable function $f$, no algorithm can solve a $\kclique{}$ instance with size $n$ in $f(k)\cdot n^{O(1)}$ time. For a $\kclique{}$ instance $G=(V=V_1\dot\cup\ldots\dot\cup V_k,E)$, we can build $k$ variables $x_1,\ldots,x_k$, letting $V_i$ be the domain of variable $x_i$ and making the edge set between $V_i$ and $V_j$ the constraint relation for $x_i$ and $x_j$. Thus,  $G$ corresponds to a 2CSP instance with $k$ variables and alphabet size at most $n$. It is easy to see that there is a clique of size $k$ in $G$ if and only if the 2CSP instance is satisfiable. Thus, we can restate the $\wone\neq\fpt$ hypothesis as follows. 

\begin{hypothesis}[$\wone\neq\fpt$]
    For any computable function $f$, no algorithm can decide whether a given 2CSP instance $\Pi=(X,\Sigma,\Phi)$ is satisfiable in $f(|X|)\cdot |\Sigma|^{O(1)}$ time.
\end{hypothesis}

The approximation version (with respect to the fraction of constraints that can be satisfiable) of $\wone\neq\fpt$ is called Parameterized Inapproximability Hypothesis (PIH).
\begin{hypothesis}[Parameterized Inapproximability Hypothesis (PIH)\cite{LRSZ20}] \footnote{Note that the original PIH in \cite{LRSZ20} states that constant approximating 2CSP parameterized by $|X|$ is $\wone$-hard. Here we use a relaxed form.}
    For any computable function $f$ and some constant $\varepsilon>0$, no algorithm can given as input a 2CSP instance $\Pi=(X,\Sigma,\Phi)$, distinguish between the following two cases in $f(|X|)\cdot |\Sigma|^{O(1)}$ time:
    \begin{itemize}
        \item (Completeness) $\Pi$ is satisfiable.
        \item (Soundness) Any assignment of $\Pi$ violates at least $\varepsilon$ fraction of constraints.
    \end{itemize}
\end{hypothesis}

We formally define \ksetcover{}, the parameterized version of the classical {\sc SetCover} problem, and the exact version of it. 
\begin{definition}[\ksetcover{}, \kexact{}]
    An instance of \ksetcover{} problem is a tuple $\Pi=(\mathcal S,U)$, where $\mathcal S$ is a collection of subsets $\{S_1,\ldots,S_n\}$ over the universe $U$, and the goal is to decide whether there are $k$ sets in $\mathcal S$, whose union is $U$. For $r>1$, the $r$-gap version of \ksetcover{} asks to distinguish between the following two cases:
    \begin{itemize}
        \item (Yes) There are $k$ sets whose union is $U$.
        \item (No) The union of any $r\cdot k$ sets is not $U$.
    \end{itemize}
    Furthermore, if in the yes case, the $k$ sets are non-intersecting, i.e., they form a partition of $U$, then we also denote this gap problem as \kexact{}.
\end{definition}

Finally, we define the $(T,m)$-set gadget, an important gadget used in reductions to \ksetcover{} and \kexact{}.

\begin{definition}[$(T,m)$-Set Gadget]
    A $(T,m)$-set gadget consists of a universe $\mathcal M$ and some of its subsets $C_1,\ldots,C_{m}$ with the following property: Every collection of at most $T$ sets out of $C_1,\overline {C_1},C_2,\overline{C_2},\ldots,C_m,\overline{C_m}$ that is a set cover for $\mathcal M$ must include both $C_i$ and $\overline{C_i}$ for some $i$.
\end{definition}

It was proved in \cite{LY94} that a $(T,m)$-set gadget can be constructed efficiently:

\begin{lemma}
\label{lem:construct-set-gadget}
    There is an algorithm that given any $T,m$, runs in time $\text{poly}(m,2^T)$ and outputs a $(T,m)$-set gadget with universe size $\text{poly}(m,2^T)$.
\end{lemma}


\section{Baby PIH}
\label{sec:baby}
\noindent In this section, we analyze the direct product construction to prove Baby PIH under $\wone \neq \fpt$.
\begin{theorem}[Main]\label{thm:baby}
    For any integer $r>1$, there is an integer $t>0$ such that for any 2CSP instance $\Pi=(X,\Sigma,\Phi)$ and its $t$-wise direct product instance $\Pi^{\odot t}=(X',\Sigma',\Phi')$:
    \begin{itemize}
        \item (Completeness) If $\Pi$ is satisfiable, then $\Pi^{\odot t}$ is satisfiable as well.
        \item (Soundness) If $\Pi$ is not satisfiable, then $\Pi^{\odot t}$ is not $r$-list satisfiable.
    \end{itemize}
\end{theorem}

Since $t$ is a constant depending solely on $r$, the number of variables in the new instance $|X'|=\binom{|X|}{t}$ depends only on $|X|$ rather than $|\Sigma|$, and the alphabet size of the new instance $|\Sigma'|\le |\Sigma|^t$ is polynomial in $|\Sigma|$. Thus for any constant $C$ and any computable function $f$, $f(|X'|)\cdot |\Sigma'|^{C}$ time is upper bounded by $g(|X|)\cdot |\Sigma|^{C\cdot t}$ time for some computable function $g$. Therefore, we have the following corollary from Theorem \ref{thm:baby}:
\begin{corollary}
    Assuming $\wone\neq\fpt$, Baby PIH is true.
\end{corollary}

Note that the completeness in~\Cref{thm:baby} follows trivially by assigning the restriction of the satisfying assignment on $X$ to each subset of variables. The main challenge is to show that when $\Pi^{\odot t}$ has a $r$-list satisfying assignment, $\Pi$ is satisfiable. To prove this, we will first work on a bipartite version of the direct product of 2CSP. 
\begin{definition}[Bipartite Direct Product Construction]
    Given a 2CSP instance $\Pi=(X,\Sigma,\Phi)$ and positive integers $a,b$, the $(a,b)$-bipartite direct product 2CSP instance $\Pi^{\odot (a,b)}=(X',\Sigma',\Phi')$ is constructed as follows.
    \begin{itemize}
        \item The variable set $X'$ consists of all $a$-sized subsets of $X$ on the left side, and all $b$-sized subsets of $X$ on the right side. With a little abuse of notation, we have $X' = \binom{X}{a} \cup \binom{X}{b}$.
        \item The domain of each variable $S \in X'$ is the set of all partial satisfying assignments for $S$ in $\Pi$. 
        \item For every $S \in \binom{X}{a}$ and $ T \in \binom{X}{b}$, we have a constraint in $\Phi'$ that checks whether $\sigma_S$ and $\sigma_T$ are consistent on the values for variables in $S \cap T$. 
    \end{itemize} 
\end{definition}

If for a 2CSP instance $\Pi$, $\Pi^{\odot t}$ is $r$-list satisfiable for $t=\max(a,b)$, by taking restrictions of its assignments on the smaller sets, it is easy to see $\Pi^{\odot(a,b)}$ is $r$-list satisfiable as well.
Thus, our goal is to show that if the bipartite instance $\Pi^{\odot(a,b)}$ is $r$-list satisfiable, then the original instance $\Pi$ is satisfiable.  

Our proof idea is borrowed from the Baby PCP theorem recently proved by Barto and Kozik \cite{BK22}. However, their theorem crucially relies on the fact that the alphabet $|\Sigma|$ is as small as a constant, which helps them to extract satisfying assignments for the smaller instance. The running time of their reduction is, therefore, $|X|^{O_{|\Sigma|}(1)}$, which is not affordable here since $|\Sigma|$ is as large as the input size. We resolve this issue by making use of the structural properties of the assignments in the direct product construction arising from the fact that they satisfy $r$-list consistency. 
If we fix some set $S$ on one side and consider its $r$ assignments, each set on the other side that intersects $S$ must have one of the $r$, which is a constant, assignments for their intersection part. We use this simple yet very useful observation when extracting the assignments in the inductive proof.

In the following, we first prove Lemma \ref{lemma:2}, which is crucial to extract list satisfying assignments for the smaller subsets. Then in Lemma \ref{lemma:3}, we analyze a special case of the bipartite direct product construction when each variable on the right (bigger) side has only one assignment, and the consistency requirement is a slightly weaker one. In Lemma \ref{lem:final-baby}, we finish the analysis of the bipartite direct product construction, from which we get~\Cref{thm:baby} as a corollary.

\begin{lemma}\label{lemma:2}
    Let $k,r,q,a,b,b'$ be integers satisfying $r,q>0, a \ge k, b \ge r\cdot b'+a$. Consider the $(a,b)$-bipartite direct product 2CSP instance based on $\Pi=(X,\Sigma,\Phi)$. Let $u$ be an $r$-sized multi-assignment for $\binom{X}{a}$ and $v$ be a $q$-sized multi-assignment for $\binom{X}{b}$. Suppose for every $S \in \binom{X}{a}, T \in \binom{X}{b}$ with $T \supseteq S$, $v(T)_{|S} \cap u(S) \ne \emptyset$. 
    Then for every $A \in \binom{X}{k}$, there is an assignment $f_A$ for $A$ such that for every $T' \in \binom{X}{b'}$, there is some $T \in \binom{X}{b}$ satisfying $T \supseteq T' \cup A$ and $v(T)_{|A} \ni f_A$.
\end{lemma}
\begin{proof}
Suppose for the sake of contradiction that there is no such $f_A$ for some set $A \in \binom{X}{k}$. In other words, for any assignment $f$ on $A$, there exists a set $T'_{f} \in \binom{X}{b'}$, such that for every $T \in \binom{X}{b}$ satisfying $T \supseteq T'_{f} \cup A$, $v(T)_{|A} \not\ni f$.

Pick an arbitrary $S \in \binom{X}{a}$ with $S \supseteq A$. This can be done since $a \ge k$.  Consider any set $T \in \binom{X}{b}$ which contains $\cup_{f \in u(S)_{|A}} (T'_f \cup S)$. Such a $T$ exists since $b \ge r\cdot b'+a$. By the assumption about $A$, $v(T)_{|A}$ does not contain any value in $u(S)_{|A}$.  This contradicts the consistency guarantee in the hypothesis of the lemma, namely that for any $T \in \binom{X}{b}$ with $T \supseteq S$, $v(T)_{|S}$ must contain some value in $u(S)$. 
\end{proof}

See Figure \ref{fig:1} for an illustration of Lemma \ref{lemma:2}.

 \begin{figure}[h]
\centering
\includegraphics[height=9cm]{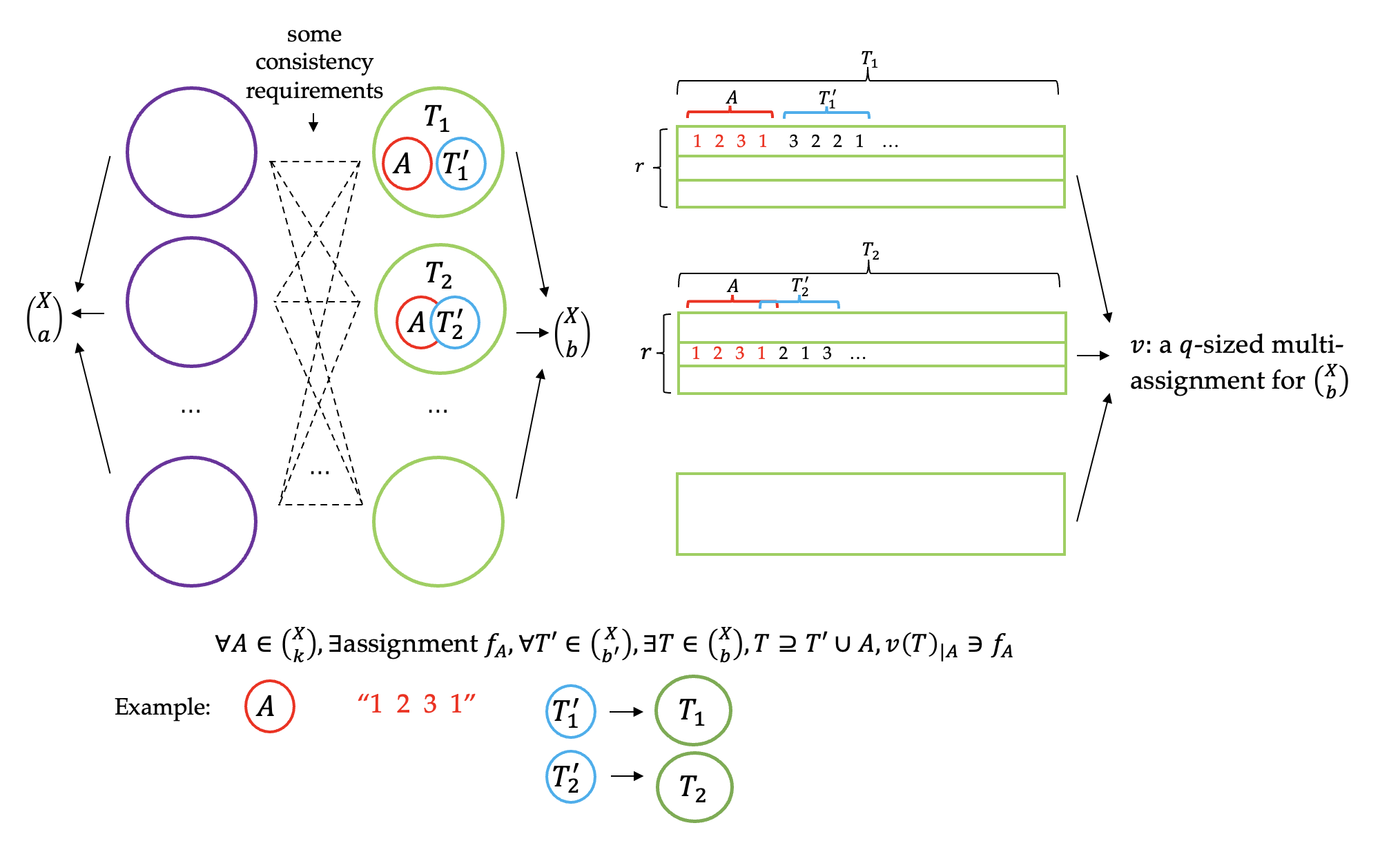}
\caption{An illustration of Lemma \ref{lemma:2}.}
\label{fig:1}
\end{figure}

\begin{lemma}\label{lemma:3}
   For any integer $r > 0$, let $a=r$ and $b=(2r)^r$. Consider the $(a,b)$-bipartite direct product 2CSP instance based on $\Pi=(X,\Sigma,\Phi)$. Let $u$ be an $r$-sized multi-assignment for $\binom{X}{a}$ and $v$ be a 1-sized assignment for $\binom{X}{b}$. Suppose for every $S \in \binom{X}{a}$ and $T \in \binom{X}{b}$ with $S \subseteq T$, $v(T)_{|S} \in u(S)$. Then there is a global satisfying assignment $\sigma$ to the 2CSP instance $\Pi$.
\end{lemma}

\begin{proof}

We apply induction on $r$. When $r=1$, we have $a=1,b=2$, and both $u,v$ are 1-sized assignments. We claim that $u$ is the desired global satisfying assignment of $\Pi$. For each constraint on $(x,y) \in X^2$, $v(\{x,y\})=(x \mapsto u(x), y \mapsto u(y))$ by our consistency guarantee. By the construction of bipartite direct product 2CSP instance, the domain for each $T \in \binom{X}{b}$ consists only of partial satisfying assignments. Thus the fact that $(x \mapsto u(x),y \mapsto u(y))$ lies in the domain of $T=\{x,y\}$ implies $u$ satisfies the constraint on $(x,y)$.

When $r>1$, the idea is to extract consistent assignments for the $(a',b')$-bipartite direct product 2CSP instance and to decrease $r$, for some $a'\le a, b' \le b$. At a high level, if for some $x \in X$, every set in $\binom{X}{a}$ has at least two different values for it under $u$, we can keep only one of them and peel the other off to decrease $r$ by 1; otherwise we can prove the unique assignments are already satisfying.

In the following, define $k=1$, and $a'=a-1,b'=(2(r-1))^{r-1}$, i.e., $a',b'$ are parameters with respect to $r-1$. It's easy to see $k \le a$ and $r\cdot b'+a \le r\cdot (2r)^{r-1}+r \le (2r)^{r}=b$. According to Lemma \ref{lemma:2}, for every $A=\{x_i\}\in \binom{X}{k}$, there is an assignment $f_A$ for $A$ such that for every $T' \in \binom{X}{b'}$, there is some $T\in \binom{X}{b}$ satisfying $T \supseteq T' \cup A$ and $v(T)_{|A}\ni f_A$. Now $v(T)$ is of size-1, so we can simply write $v(T)_{|A}=f_A$.

Consider the following two cases:

\medskip\noindent
\textbf{Case 1:} For some $A=\{x_i\}$, Lemma \ref{lemma:2} holds for different assignments $p,q$. In other words, for every $T' \in \binom{X}{b'}$, there are $T_1,T_2 \in \binom{X}{b}$ satisfying $T_1,T_2 \supseteq T' \cup \{x_i\}$ and $v(T_1)_{|\{x_i\}}=p,v(T_2)_{|\{x_i\}}=q$. 
    
    Since for $r\ge 2$ we have $a=r\le (2(r-1))^{r-1}=b'$, for each $S \in \binom{X}{a}$ containing $x_i$, we can pick an arbitrary $T'\supseteq S$ and consider the corresponding sets $T_1,T_2$ above. By the consistency assumption between $(S,T_1)$ and between $(S,T_2)$ in Lemma \ref{lemma:3}, we can infer $u(S)_{|\{x_i\}}\ni p,q$.

    We construct new assignments $u',v'$ for the $(a',b')$-bipartite direct product 2CSP instance of $\Pi$, such that they still meet the consistency requirements in Lemma \ref{lemma:3}, and the size of $u'$ is at most $r-1$. For each $S' \in \binom{X}{a'}$, $u'(S')$ will be inherited from $u(S)$ for some $S \in \binom{X}{a}$ satisfying $S \supseteq S'\cup\{x_i\}$. Similarly for each $T' \in \binom{X}{b'}$, $v'(T')$ will be inherited from $v(T)$ for some $T \in \binom{X}{b}$ satisfying $T \supseteq T'\cup\{x_i\}$. 

    Suppose there is a arbitrary fixed order of all variables in $X$. We construct $S$ from $S'$ as follows, and output $u'(S')=\{\sigma \in u(S)|\sigma(x_i)=p\}_{|S'}$:
    \begin{itemize}
        \item If $x_i \not\in S'$, let $S=S'\cup \{x_i\}$.
        \item If $x_i \in S'$, $S$ is obtained by adding the lexicographical smallest element not in $S'$ to $S'$.
    \end{itemize}
    Note that we only keep the assignments in $u(S)$ whose restriction on $\{x_i\}$ is $p$, and discard those whose restriction on $\{x_i\}$ is $q$. Thus $u'$ is of size at most $r-1$ as desired.

    For each $T' \in \binom{X}{b'}$, we first construct $T''\in \binom{X}{b'}$ as follows:
    \begin{itemize}
        \item If $x_i \notin T'$, simply let $T''=T'$.
        \item If $x_i \in T'$, $T''$ is obtained by adding the lexicographical smallest element not in $T'$ to $T'$, and delete $x_i$.
    \end{itemize}
    After that, we find a $T \in \binom{X}{b}$ satisfying $T \supseteq T'' \cup \{x_i\}$ and $v(T)_{|\{x_i\}}=p$, and output $v'(T')=v(T)_{|T'}$.

    By our construction, $S' \subseteq T'$ implies $S \subseteq T$, so the new instance still meets the required consistency requirements, and the induction proceeds.

\medskip\noindent\textbf{Case 2:} Suppose for $A$ being each $\{x_i\}$, Lemma \ref{lemma:2} holds for a unique assignment $z_i$. 

    We claim that $\left(x_i \mapsto z_i\right)_{x_i \in X}$ is a global satisfying assignment. 
    Indeed, consider an arbitrary 
     constraint $\psi$ between variables $(x_i,x_j) \in X^2$. Applying Lemma~\ref{lemma:2} to the choice $A=\{x_i,x_j\}$, we know there is an assignment $f_A$ such that for every $T' \in \binom{X}{b'}$, there is some $T \in \binom{X}{b}$ satisfying $T \supseteq T' \cup A$ and $v(T)_{|A}=f_A$. Since $v(T)$ satisfies all constraints within $T$, this means that $f_A$ must satisfy the constraint $\psi$. By the uniqueness assumption of this case, we must have $f_A(x_i)=z_i$ and $f_A(x_j) = z_j$, which means that the assignment $(z_i,z_j)$ satisfies the constraint $\psi$ between $(x_i,x_j)$. 
\end{proof}

We will now finish the analysis of the bipartite direct product construction.
\begin{lemma}
\label{lem:final-baby}

For integers $r,q>0$, let $a=(2r)^{r+2q}, b=(2r)^{r+2q+2}$. Consider the $(a,b)$-bipartite direct product 2CSP instance based on $\Pi=(X,\Sigma,\Phi)$. Let $u$ be an $r$-sized multi-assignment for $\binom{X}{a}$ and $v$ be a $q$-sized multi-assignment for $\binom{X}{b}$. If $u,v$ list-satisfy all bipartite constraints, then there is a global satisfying assignment $\sigma$ to  the 2CSP instance $\Pi$.

\end{lemma}

\begin{proof}
We apply induction on $q$, the size of the right multi-assignment. When $q=1$, we have $a=(2r)^{r+2}\ge r$ and $b=(2r)^{r+4} \ge (2r)^r$. We can extract consistent satisfying assignments for the $(r,(2r)^r)$-bipartite direct product 2CSP instance, and invoke Lemma \ref{lemma:3} to prove there is a global satisfying assignment.

When $q>1$, we categorize discussions based on whether the left multi-assignment $u$ satisfies a certain property. Either we can still extract consistent assignments for the smaller $(a',b')$-bipartite direct product 2CSP instance while decreasing $q$ and leaving $r$ unchanged, or we can build multi-assignments that satisfy requirements in Lemma \ref{lemma:3}, and therefore directly invoke that lemma and stop the induction.

Define $k=(2r)^r$. The parameters $a',b'$ with respect to $q-1$ would be $a'=(2r)^{r+2q-2}$ and $b'=(2r)^{r+2q}$. In our case-analysis we will use the following inequalities, which we first prove here.

\begin{itemize}
    \item $b \ge r\cdot b'+a$.  Indeed 
    \[
    \begin{aligned}
        r\cdot b'+a & = r\cdot (2r)^{r+2q}+(2r)^{r+2q} \\
        & \le 2\cdot (2r)^{r+2q+1} \\
        & \le (2r)^{r+2q+2} =b.
    \end{aligned}
    \]
    \item $a \ge (r+1)\cdot (a'+k)$. Indeed
    \[
    \begin{aligned}
        (r+1)\cdot (a'+k) & = (r+1)\cdot ((2r)^{r+2q-2}+(2r)^r) \\
        & \le (2r)\cdot (2 \cdot (2r)^{r+2q-2}) \\
        & \le (2r)^{r+2q} = a.
    \end{aligned}
    \]
\end{itemize}

Now we consider the following two cases based on $u$. The criterion is whether $u$ satisfies a property which is reminiscent of the result of Lemma \ref{lemma:2}, except in Lemma \ref{lemma:2} the property is for $v$ while here we check it for $u$.

\begin{enumerate}
    \item Suppose there exists $A \in \binom{X}{k}$ such that for every assignment $f_A$ on $A$ and every $S' \in \binom{X}{a'}$, there exists $S \in \binom{X}{a}$ satisfying $S \supseteq S' \cup A$ and $u(S)_{|A} \not\ni f_A$.

    Given $a \ge k$ and $b \ge r\cdot b'+a$, we can apply Lemma \ref{lemma:2} to know that there is an assignment $f_A$ on $A$, such that for every $T' \in \binom{X}{b'}$, there is $T \in \binom{X}{b}$ satisfying $T \supseteq T'\cup A$ and $v(T)_{|A}\ni f_A$.

    We therefore build multi-assignments $u',v'$ for the smaller $(a',b')$-bipartite 2CSP instance as follows.

    For every $S' \in \binom{X}{a'}$, we pick the set $S$ guaranteed by the assumption of this case, and define $u'(S')=u(S)_{|S'}$. Note that $u'$ still has size at most $r$.

    For every $T' \in \binom{X}{b'}$, we pick the set $T$ guaranteed by Lemma \ref{lemma:2}, and define $v'(T')=\{\sigma \in v(T)| \sigma_{|A} \neq f_A\}_{|T'}$. By the assumption, there exists $S \in \binom{X}{a}$ satisfying $S \supseteq A$ and $u(S)_{|A}\not\ni f_A$. Thus to be consistent with $u(S)_{|A}$, we can conclude for every $T \in \binom{X}{b}$ containing $A$, $v(T)_{|A}$ should also contain some value other than $f_A$. Therefore, our constructed $v'$ is non-empty and has size at most $q-1$.

    It's also easy to see $u',v'$ still satisfy list consistency, since in $v'$ we only discard the assignments whose restriction on $A$ equals to $f_A$, which are not consistent with any $u'$.
    
    \item Suppose for every $A \in \binom{X}{k}$, there is an assignment $f_A$ for $A$ and a set $S' \in \binom{X}{a'}$, such that for every $S \in \binom{X}{a}$ satisfying $S \supseteq S' \cup A$, we have $u(S)_{|A}\ni f_A$.

    In this case we can construct $r$-sized multi-assignment $u'$ and $1$-sized assignment $v'$ for left and right part of the $(r,k=(2r)^r)$-bipartite direct product 2CSP of $\Pi$ respectively, that meets the requirements of Lemma \ref{lemma:3}. Furthermore, $u',v'$ are built purely based on $u$.

    For every $B \in \binom{X}{r}$, we define
    \[
    u'(B)=\bigcup_{A \in \binom{X}{k},A \supseteq B} (f_A)_{|B}
    \]
    For every $A \in \binom{X}{k}$, we define $v'(A)=f_A$.
    
    We first claim the size of $u'$ is at most $r$. Suppose it is not, there are $r+1$ sets $A_1,\ldots,A_{r+1} \in \binom{X}{k}$ with $f_{A_i}$ all different. Let $S'_1,\ldots,S'_{r+1}\in\binom{X}{a'}$ be the corresponding sets guaranteed in the assumption of this case. Consider a $S \in \binom{X}{a}$ which contains $\cup_{i=1}^{r+1}(S'_i \cup A_i)$. Such $S$ exists since $a \ge (r+1)\cdot (a'+k)$. Thus, $u(S)_{|A}$ contains $(r+1)$ different values, contradicting the fact that $u$ is of size $r$.

    It's easy to see $u',v'$ meets the requirement of Lemma \ref{lemma:3} by the definition of $u'$. Thus using Lemma \ref{lemma:3}, there is a satisfying assignment to the original instance $\Pi$. \qedhere
\end{enumerate}
\end{proof}

\noindent
Our main result, Theorem \ref{thm:baby}, now follows immediately from~\Cref{lem:final-baby}.

\begin{proof}[Proof of Theorem~\ref{thm:baby}]
    Given an integer $r$, we pick $t=(2r)^{r+2r+2}$, and prove that if a 2CSP instance $\Pi=(X,\Sigma,\Phi)$ is satisfiable, then so is $\Pi^{\odot t}$; if $\Pi$ is not satisfiable, then $\Pi^{\odot t}$ is not $r$-list satisfiable.

    The completeness case is easy: let $\sigma$ be a satisfying assignment for $\Pi$, then we can assign each set $S \in \binom{X}{t}$ the function that maps any $x \in S$ to $\sigma(x)$.

    For soundness case, suppose $\Pi^{\odot t}$ is $r$-list satisfiable by an assignment $\sigma$, we take $a=(2r)^{r+2q}, b=(2r)^{r+2q+2}$ and build multi-assignments $u,v$ for the left and right part of the $(a,b)$-bipartite direct product 2CSP instance $\Pi^{\odot(a,b)}$. For each set $S' \in \binom{X}{a}$, we pick an arbitrary $S \in \binom{X}{t}$ with $S \supseteq S'$, and define $u(S')=\sigma(S)_{|S'}$. Similarly for each $T'\in \binom{X}{b}$, we pick an arbitrary $T \in \binom{X}{t}$ with $T \supseteq T'$, and define $v(T')=\sigma(T)_{|T'}$. It's easy to see $u$ and $v$ are $r$-list consistent. Thus by Lemma \ref{lem:final-baby}, $\Pi$ is satisfiable.
\end{proof}
\section{Average Baby PIH}
\label{sec:avg}

\noindent Let us recall the average Baby PIH conjecture. 
\begin{hypothesis}[Average Baby PIH]
    Given a 2CSP instance $\Pi=(X,\Sigma,\Phi)$, we say a multi-assignment $\sigma$ $r$-average-list satisfies $\Pi$ if $\sum_{x \in X}|\sigma(x)|\le r\cdot |X|$, and for every constraint $\phi_i$, there exists $u \in \sigma(w_{i,1})$ and $v \in \sigma(w_{i,2})$, such that $(u,v) \in R_i$.

    Average Baby PIH states that for any constant $r>1$ and any computable function $f$, no algorithm can given as input a 2CSP instance $\Pi=(X,\Sigma,\Phi)$, distinguish between the following two cases in $f(|X|)\cdot |\Sigma|^{O(1)}$ time:
    \begin{itemize}
        \item (Completeness) $\Pi$ is satisfiable.
        \item (Soundness) $\Pi$ is not $r$-average-list satisfiable.
    \end{itemize}
\end{hypothesis}

\subsection{A Counter Example for Direct Product Construction}

We use the following counter example to show that the Direct Product construction does not give us Average Baby PIH. Specifically, for any $t>0$ and $\varepsilon>0$, there exists an 2CSP instance which is not satisfiable but its $t$-wise direct product is $(1+\varepsilon)$-average-list satisfiable.

\begin{example}\label{ex:1}
    The 2CSP instance $\Pi=(X,\Sigma,\Phi)$ is defined as follows.
    \begin{itemize}
        \item $X=\{x_1,\ldots,x_n\}$.
        \item $\Sigma_{x_1}=\{2,\ldots,n\}$, and for every $i \in \{2,\ldots,n\}$, $\Sigma_{x_i}=\{1\}$.
        \item For every $i \in \{2,\ldots,n\}$, there is a constraint $\phi_i=\langle w_i,R_i\rangle \in \Phi$, where
        \begin{itemize}
            \item $w_i=(x_1,x_i)$.
            \item $R_i=\{(j,1) | j \ne i\}$.
        \end{itemize}
    \end{itemize}
    \end{example}
    $\Pi$ is not satisfiable since no value for $x_1$ can satisfy all constraints. Specifically, for $i \in \{2,\ldots,n\}$, $x_1=i$ would violate constraint $\phi_i$.

    However, for every integer $t>0$, $\Pi^{\odot t}$ can be list satisfied by the following multi-assignment $\sigma$. For every $S \in \binom{X}{t}$, 
    \begin{itemize}
        \item if $x_1 \notin S$, define $\sigma(S)$ to be the single assignment that maps every $x_i \in S$ to 1;
        \item if $x_1 \in S$, then for every $2 \le j \le 2t$ with $x_j \notin S$, let $\sigma(S)$ contain an assignment that maps $x_1$ to $j$, and maps everything else to 1.
    \end{itemize}
    
    It's easy to see $\sigma$ satisfies all constraints induced by $S$. It remains to show that $\sigma$ is consistent on every different $S,T \in \binom{X}{t}$. 
    
    If $x_1 \notin S$ or $x_1 \notin T$, $\sigma(S)$ and $\sigma(T)$ are trivially consistent since every variable in $S \cap T$ is always assigned 1. Now suppose $x_1 \in S\cap T$, we have $|S \cup T| \le 2t-1$. By Pigeonhole Principle, there is a variable $x_j$ with $j \le 2t$, which is neither in $S$ nor in $T$. Thus the assignment that maps $x_1$ to $j$ and everything else to $1$ appears in both $\sigma(S)$ and $\sigma(T)$, proving that $\sigma$ list-satisfies $\Pi^{\odot t}$. The total size of the lists is \[
    \begin{aligned}
        \frac{1}{\binom{|X|}{t}}\sum_{x \in X} |\sigma(x)| & \le \frac{\binom{|X|-1}{t}}{\binom{|X|}{t}}\cdot 1+\frac{\binom{|X|-1}{t-1}}{\binom{|X|}{t}}\cdot 2t \\ 
        & = \left(1-\frac{t}{|X|}\right)\cdot 1+\frac{t}{|X|}\cdot 2t \\
        & \le 1+\frac{2t^2}{|X|},
    \end{aligned}\]
    which is smaller than $1+\varepsilon$ when $|X|$ goes to infinity.

\subsection{Towards the Inapproximability of \kexact}
In this subsection, we show that a slightly strengthened version of Average Baby PIH, namely, Average Baby PIH with rectangular relations, implies constant inapproximability of \kexact{} problem. The latter, to our best knowledge, is currently not known under $\wone \neq \fpt$. Formally, we have the following theorem:
\begin{theorem}\label{thm:exactcover}
    Suppose Average Baby PIH holds even when the 2CSP instance has rectangular relations (see Definition \ref{def:rect}), then for any constant $r \ge 1$ and any function $f$, no algorithm can approximate \kexact{} problem within factor $r$ in $f(k)\cdot n^{O(1)}$ time. Specifically, no algorithm can given a \ksetcover{} instance $\Pi=(\mathcal S,U)$ with size $n$, distinguish between the following two cases in $f(k)\cdot n^{O(1)}$ time:
    \begin{itemize}
        \item There exists $k$ sets in $\mathcal S$ which form a partition of $U$.
        \item The union of any $r\cdot k$ sets in $\mathcal S$ is not $U$.
    \end{itemize}
\end{theorem}
\begin{proof}
We reduce a 2CSP instance $\Pi=(X,\Sigma,\Phi)$ with rectangular relations to a \ksetcover{} instance $\Pi'=(\mathcal S,U)$ with $k=|X|$ in the following way. For every $(x,v) \in X \times \Sigma$, we build a set $S_{x,v}$. For each constraint $\phi_i=\langle (w_{i,1},w_{i,2}),R_i\rangle$ in $\Phi$, let $\Gamma_i$ be the underlying set in the rectangular relation $R_i$ and let $\pi_i,\sigma_i:\Sigma\to \Gamma_i$ be the underlying mappings. We build a $(r \cdot k,|\Gamma_i|)$-set gadget $(\mathcal M^{(i)},C_1^{(i)},\ldots,C_m^{(i)})$. For every $(a,b) \in R_i$, we add the set $C_{\pi_i(a)}^{(i)}$ to $S_{w_{i,1},a}$, and add the set $\overline{C_{\sigma_i(b)}^{(i)}}$ to $S_{w_{i,2},b}$. Let the final universe $U$ be the union of every $\mathcal M^{(i)}$.


In the completeness case, let $\sigma:X \to \Sigma$ be a satisfying assignment of $\Pi$. It is easy to see the $k$ sets $\{S_{x,\sigma(x)}\}_{x \in X}$ cover each element of the universe exactly once.

In the soundness case, let $\mathcal S' \subseteq S$ be a collection of sets that covers $U$. Assuming $|\mathcal S'| \le r \cdot k$, we claim the multi-assignment $\sigma$, which maps $x \in X$ to $\{v \in \Sigma \mid S_{x,v} \in \mathcal S'\}$, is a list satisfying assignment of $\Pi$. For every constraint $\phi_i=\langle (w_{i,1},w_{i,2}),R_i\rangle$ with $\Gamma_i$ being the image of the rectangular mapping, by the property of $(r\cdot k,|\Gamma_i|)$-set gadget and the fact that $|\mathcal S'|\le r \cdot k$, $\mathcal S'$ must include $C_{j}^{(i)}$ and $\overline{C_j^{(i)}}$ for some $1 \le j \le |\Gamma_i|$, which implies that $\mathcal S'$ includes $S_{w_{i,1},a}$ and $S_{w_{i,2},b}$ for some $(a,b) \in R_i$, and thus $\phi_i$ is list satisfied.


From \Cref{lem:construct-set-gadget}, a $(r\cdot k,|\Gamma_i|)$-set gadget can be constructed in time $\text{poly}(|\Gamma_i|,2^{r \cdot k})$. The whole reduction runs in FPT time while preserving $k=|X|$. Thus, an $f(k)\cdot n^{O(1)}$ time algorithm for $r$-approximating \kexact{} would give an $f(k)\cdot |\Sigma|^{O(1)}$ algorithm for the $r$-average-list satisfiability of 2CSP with rectangular relations, contradicting the strengthened version of Average Baby PIH.
\end{proof}

We remark that the direct product construction does have rectangular relations, although it does not directly give us Average Baby PIH, in view of Example~\ref{ex:1}.

\section{Discussion and Open Problems}
\label{sec:discuss}

In this concluding section, we speculate on some possible avenues to attack Average Baby PIH or even PIH itself.

\subsection{Average Baby PIH from Clique Hardness?}

In \cite{Lin21}, the author proved that even constant approximating \kclique{} is $\wone$-hard. In \cite{KK22,CFL+23}, the inapproximability ratio was improved to $k^{o(1)}$. Using CSP words, their results can be described as the following theorem:

\begin{theorem}[\cite{KK22,CFL+23}]
    Assuming $\wone\neq\fpt$, no algorithm can, given a 2CSP instance $\Pi=(X,\Sigma,\Phi)$, distinguish between the following two cases in $f(|X|)\cdot |\Sigma|^{O(1)}$ time, for any computable function $f$:
    \begin{itemize}
        \item $\Pi$ is satisfiable.
        \item The constraints induced by any $|X|^{1-o(1)}$ variables are not simultaneously satisfiable.
    \end{itemize}
\end{theorem}
 Thus, it is natural to ask whether we can get Average Baby PIH by applying direct product construction to a 2CSP instance with the above ``clique-like" soundness? More formally, we have the following open question:

\begin{question}
Is it true that for any integer $r>1$, there is an integer $t>0$, such that for any 2CSP instance $\Pi=(X,\Sigma,\Phi)$, if $\Pi^{\odot t}$ is $r$-average-list satisfiable, then there are $|X|^{1-o(1)}$ variables in $\Pi$ such all constraints amongst them are simultaneously satisfiable?
\end{question}

Note that the counterexample for proving average baby PIH using the direct product construction (Example \ref{ex:1}) does not apply here, since there are $|X|-1$ variables in $\Pi$ that are simultaneously satisfiable.

\subsection{PIH via Direct Product Testing Theorems?}

Our constructed instance in proving Baby PIH is reminiscent of the direct product testing, which has been studied in a recent line of work \cite{DS14, DK17, DFH19}. These results have shown that if the $t$-sized subsets of variables have good local consistency, then there is a global function which agrees with most of the subsets. Formally, we have the following theorem from \cite{DFH19}.

\begin{theorem}
\label{thm:dir}
    There is an absolute constant $C>1$, such that for any $\alpha,\beta \in (0,1)$ with $\alpha+\beta \le 1$, there exists a constant $Q(\alpha,\beta)$, such that given any $k,t,m$ with $k \ge C \cdot t$ and $\alpha t \le m \le (1-\beta) t$ and any finite alphabet $\Sigma$ , we have the following.

    Let $\mathcal F=\{f_S:S \rightarrow \Sigma \mid S \in \binom{[k]}{t}\}$ be an ensemble of functions, one for every size-$t$ subset of $[k]$. Let $\mathcal D(m)$ be the distribution as follows:
    \begin{itemize}
        \item Choose $I \in \binom{[k]}{m}$ uniformly at random.
        \item Choose $A,B$ from the set $\{X \mid X \in \binom{[k]}{t}, X \supseteq I\}$ uniformly at random.
    \end{itemize}

    Suppose the following holds:
    \[
    \Pr_{A,B \sim \mathcal D(m)}[{f_A}_{|A \cap B} = {f_B}_{|A \cap B}] \ge 1-\varepsilon,
    \]
    then there exists a global function $g: [k] \rightarrow \Sigma$ such that
    \[
    \Pr_{A \sim \binom{[k]}{t}}[f_A = g_{|A}] \ge 1-Q(\alpha,\beta) \cdot \varepsilon.
    \]
\end{theorem}

By setting the alphabet of each $f_S$ to be the set of all partial satisfying assignments for $S$ and adding the consistency checks according to the above theorem, one may wonder whether we can ``extract'' a large clique from the $1-Q(\alpha,\beta)\cdot \varepsilon$ fraction of subsets that are globally consistent. Formally, let $\mathcal S=\{A | A \in \binom{[k]}{t} , f_A=g_{|A}\}$ as in \Cref{thm:dir}, can we prove there exists a subset $T \subseteq [k]$ of size $\ge k^{1-o(1)}$, such that the following holds?
\begin{itemize}
    \item For every $(i,j) \in \binom{T}{2}$, there exists $A \in \mathcal S$ with $(i,j) \subseteq A$.
\end{itemize}

However, this might not be true if we only use the size bound on $\mathcal S$ and not the additional structures. Consider the following counter-example:

\begin{example}
    Mark each pair $(i,j) \in \binom{[k]}{2}$ as 1 independently with probability $1-\gamma$ with $\gamma$ to be determined. Take $\mathcal S'$ to be the collection of $t$-sized subsets of $[k]$, with all pairs in it marked as 1: $\mathcal S':=\{A \mid A \in \binom{[k]}{t}, \forall (i,j) \subseteq A, (i,j)\text{ is marked as 1}\}$.

    The probability that a $t$-sized subset belongs to $\mathcal S'$ is $(1-\gamma)^{\binom{t}{2}}$, which can be made arbitrarily close to 1 when we set $\gamma=\gamma(t)$ to be some small enough constant depending on the constant $t$.

    However, it was known that (see e.g. \cite{Grimmett_McDiarmid_1975}) the maximum clique size in this Erd\H{o}s-R\'enyi graph is only $2 \log k/\log (1/(1-\gamma))$, far smaller than $k^{1-o(1)}$.
\end{example}

This example suggests that, in order to potentially prove PIH from direct product testing theorems, one may need to analyze the structure of the collection $\mathcal S$. We leave this as an interesting future direction:

\begin{question}
    Can we prove PIH under $\wone \ne \fpt$ using some appropriate form of direct product testing theorems?
\end{question}

\section*{Acknowledgment}
We thank Shuangle Li for pointing out a mistake in the previous proof of \Cref{thm:exactcover}.

\bibliographystyle{alpha}
\bibliography{main}

\end{document}